%% file: FiniteBlockLen.tex
\documentclass[10pt, conference,onecolumn]{IEEEtran}  

\input{common_defs.tex}

\newif\ifFullProofs
\FullProofstrue 

\begin{document}

\title{Achievable and Converse bounds over a general channel and general decoding metric}

\author{Nir~Elkayam ~~~~~~~Meir~Feder \\
        Department of Electrical Engineering - Systems\\
        Tel-Aviv University, Israel \\
        Email: nirelkayam@post.tau.ac.il, meir@eng.tau.ac.il
}

\maketitle

\subsection*{\centering Abstract}
\textit{
Achievable and converse bounds for general channels and mismatched decoding are derived. The direct (achievable) bound is derived using random coding and the analysis is tight up to factor 2. The converse is given in term of the achievable bound and the factor between them is given. This gives performance of the best rate-R code with possible mismatched decoding metric over a general channel, up to the factor that is identified. In the matched case we show that the converse equals the minimax meta-converse of Polyanskiy \etal\ \cite{polyanskiy2010channel}.
}

\section{Introduction}\label{sec:intro}

Much attention has been drawn in the last few years to the analysis of coding performance in the finite block length regime, especially the memoryless case, where the normal approximation provides a means to approximate the coding rate as a function of the block-length and fixed error probability.

Information spectrum methods were used to derive a general formula for the channel capacity \cite{DBLP:journals/tit/VerduH94}, \cite{somekh2013general} in both the matched and the mismatched cases; however, these methods were not implemented directly to the finite block length regime or the error exponent (large deviation) regime.

In this paper we take a somewhat general approach, and evaluate the general (channel and decoding metric) random coding situation. The \emph{pairwise error probability} is crucial to the evaluation of the random coding error probability and is usually defined by considering ties as errors. This was already done in \cite{ElkayamEE2014}, but a more refined analysis of how to resolve ties, given in this paper, allows us to define the \emph{exact} pairwise error probability that must be \textbf{randomized} and can be regarded as a continuous generalization of the pairwise error probability, even in the discrete case.

Applying the formulas from \cite{ElkayamEE2014} to the randomized pairwise error probability allows us to provide tight bounds that are tighter than previously known bounds \cite{haim2013importance} (RCU, RCU*, and DT bounds). Formulas for the error probability, error exponent, and their connection as a function of the rate are given.

A ``converse'' result for the random coding situation is proved to be a converse for general deterministic codes (also in the mismatched case). In the matched case we show that the converse equals the minimax meta converse of Polyanskiy \etal \cite{polyanskiy2010channel}.

Combining these observations, the main contribution of this paper is as follows: In any general channel there are two possible situations. In the first, the performance of random coding can become very close to the best optimal code at a given rate (up to a constant factor); the second is a straight line regime, which means that the derivative of the random coding error exponent is close to -1 and random coding cannot do better than this. This phenomenon is known for DMCs, where the random coding error exponent is known to be tight for rates above the critical rate according to the sphere packing bound.

The paper is organized as follows: In section \ref{section:background} we present the background for the results in the paper. In section \ref{section:analysis} we provide a new analysis of the general random coding performance with a randomized maximum metric decoder that we define. In section \ref{section:general_formula} we use the formulas derived in \cite{ElkayamEE2014} to get a general expression for the error probability and error exponents in terms of the cumulative distribution function (CDF) of the log of the pairwise error probability. In section \ref{section:general_converse} we prove the converse for the general case and show the equivalence to the minimax meta-converse. In \ref{section:discussion} we discuss the results and present further research directions that are on-going investigations.

\section{Notation and Background} \label{section:background}
\subsection{Notation}
This paper uses bold lower case letters (\eg, $\bx$) to denote a particular value of the corresponding random variable denoted in capital letters (\eg, $\bX$). Calligraphic fonts (\eg, $\cX$) represent a set. $W(\by|\bx)$ denotes a general channel from $\cX$ to $\cY$. Throughout this paper $\log$ will be defined to base $e$ and rates are expressed in $nats$. $\PR{A}$ will denote the probability of the event $A$.


\subsection{Random Coding}
A \textbf{random code} with $M=e^R$ codewords and \textbf{prior distribution} $Q(\bx)$ on $\cX$ is:
\begin{align}
\cC = \BRAs{\bx_1, \bx_2, ..., \bx_{M} } \subset \cX
\end{align}
where each codeword $\bx_i$ is drawn independently according to the distribution $Q(\bx)$. The \textbf{encoder} maps the input message $k$, drawn uniformly over the set $\BRAs{1, ..., M}$, to the appropriate codeword, \ie, $\bx_k$. A \textbf{decoding metric} is a function $m: \cX \times \cY \to \mathbb{R}$. The \textbf{random decoder} associated with the decoding metric $m$ is the \textbf{random function} $\cD_m: \cY \to \BRAs{1, ..., M}$ defined by:

\begin{equation}\label{RandomizeDecoder:def1}
\cD_m(\by) = \hat{i} = argmax_{1 \leq i \leq M} \BRA{m(\bx_i,\by)},
\end{equation}
where the decoder randomly chooses, uniformly, between all the competing words that have the same maximal metric value.

\begin{definition}
The \textbf{average error probability} associated with the decoder $\cD_m$ over the channel $W(\by|\bx)$ is denoted by $\bar{P}_{e,m,W}\BRA{R}$. The error probability captures the randomness in the codebook selection, the transmitted message, the channel output, and the decoder randomness.
\end{definition}

\begin{remark}
When the decoding metric is $m(\bx,\by) = \log W(\by|\bx)$, the error probability is minimized and the decoder is called the \emph{Maximum Likelihood} decoder. The \textbf{mismatch} case is when the metric is not matched to the channel and may result in performance degradation.
\end{remark}

\subsection{Binary Hypothesis testing and the Meta-Converse}

Recall some general (and standard) definition about the optimal performance of a binary hypothesis testing between two probability measures $P$ and $Q$ on $W$:
\begin{equation}\label{binary_hypothsis:beta}
  \beta_{\alpha}\BRA{P,Q} = \MIN{\substack{P_{Z|W} :\\ \sum_{w\in W}P(w)P_{Z|W}(1|w) \geq \alpha} } \sum_{w\in W}Q(w)P_{Z|W}(1|w),
\end{equation}
where $P_{Z|W}:W \rightarrow \BRAs{0,1}$ is any randomized test. The minimum is guaranteed to be achieved by the
Neyman Pearson lemma. Thus, $\beta_{\alpha}\BRA{P,Q}$ gives the minimum probability of error under hypothesis $Q$ if the probability of error under hypothesis $P$ is not larger than $1-\alpha$. $\beta$ is the \textbf{power} at \textbf{significance level} $1-\alpha$.

\begin{remark}
  The optimal test is:
  \begin{equation*}
    P_{Z|W} = \Ind{ \log\BRA{\frac{P(w)}{Q(w)} } > \lambda} + \delta\cdot\Ind{ \log\BRA{\frac{P(w)}{Q(w)} } = \lambda},
  \end{equation*}
  where $\lambda, \delta$ are tuned so that $\sum_{w\in W}P(w)P_{Z|W}(1|w) = \alpha$
\end{remark}

Polyanskiy \etal\ \cite{polyanskiy2010channel} proved the following general converse result for the average error probability:
\begin{theorem}
  For any code with $M$ equiprobable codewords:
  \begin{equation}\label{MetaConverse:1}
    \epsilon \geq \inf_{Q(\bx)}\sup_{Q(\by)} \beta_{1-\frac{1}{M}}\BRA{Q(\bx)Q(\by), Q(\bx)W(\by|\bx)}
  \end{equation}
  \begin{equation}\label{MetaConverse:2}
    \frac{1}{M} \geq \inf_{Q(\bx)}\sup_{Q(\by)} \beta_{1-\epsilon}\BRA{Q(\bx)W(\by|\bx), Q(\bx)Q(\by)}.
  \end{equation}
\end{theorem}

The first version gives the lower bound on the error probability in terms of the rate and the second version gives the lower bound on the rate in terms of the error probability. Using equation \eqref{MetaConverse:1} and instantiating $Q(\by)$, they proved that lots of other known lower bounds on the error probability can be derived from this converse and hence the name: \textbf{The Meta Converse}.

\section{Analysis of the randomized decoder performance with random coding} \label{section:analysis}

Breaking ties should be done arbitrarily between the codewords that receive maximum metric value. We propose to add a ``small'' dither to each metric so that the probability of ties is 0. The dither should be small enough so that the order between different symbols in the original metric remains the same in the randomized metric.

We proposed the following interpretation of the random decoder. Given $\by$, the received word, and the metric $m(\bx,\by)$ define the randomized metric of the $k$-th codeword as follows:

\begin{align}\label{pairwise:def}
  \xi_{k,y} &= p_{e,m,x_k,y} \notag \\
  &= Q\BRA{m(\bX,\by) > m(\bx_k,\by)} + \mathcal{U}_k \cdot Q\BRA{m(\bX,\by) = m(\bx_k,\by)}
\end{align}

where $Q(\bx)$ is the prior and $\bx_k$ is the codeword associated with the $k$-th symbol and $\mathcal{U}_k$ is uniform over the unit interval $[0,1]$ and independent for each metric. Notice that this metric preserves the ordering in the original metric. When using this random metric we decode for the \emph{minimum} metric, probability of ties is 0 and the performance is equivalent to that of the randomized decoder defined in \eqref{RandomizeDecoder:def1}.

In the following proposition we summarize the properties of this random metric:
\begin{proposition} \label{prop:propeties_pe}\mynewline
(1) For any $\bx_1, \bx_2, \by$:
\begin{equation}\label{RandomMetric:Order}
  m(\bx_1,\by) > m(\bx_2,\by) \Rightarrow p_{e,m,x_1,y} < p_{e,m,x_2,y}
\end{equation}
(2)
\begin{align*}
          Q\BRA{m(\bX,\by) > m(\bx,\by)} &\leq p_{e,m,x,y} \\
                                         &\leq Q\BRA{m(\bX,\by) \geq m(\bx,\by)}
\end{align*}
(3)
\begin{equation}\label{RandomMetric:Uniform}
  \PR{p_{e,m,\bX,y} < u} = u, u \in [0,1]
\end{equation}
where the probability is with respect to $Q(\bx)$ and the uniform variable $\mathcal{U}$, \ie\ $p_{e,m,\bX,y}$ is uniform on $[0,1]$.
\end{proposition}
Notice that property \eqref{RandomMetric:Order} is exactly the order preserving of the random metric we proposed.
\begin{proof}
The first 2 properties are obvious. To prove the third property, notice that for each $u$, there exist $\bx$ and $\tau$ such that: $u=Q\BRA{m(\bX,\by) > m(\bx,\by)} + \tau \cdot Q\BRA{m(\bX,\by)=m(\bx,\by)}$; then by using \eqref{RandomMetric:Order}:
\begin{align*}
 \PR{p_{e,m,X,y} < u} &= \sum_{\bx':m(\bx',\by) > m(\bx,\by)}Q(\bx') \\
                      &+ \PR{ \mathcal{U} < \tau} \cdot \sum_{\bx':m(\bx',\by) = m(\bx,\by)}Q(\bx') \\
  &= Q\BRA{m(\bX,\by) > m(\bx,\by)} \\
  &+ \tau \cdot Q\BRA{m(\bX,\by)=m(\bx,\by)} \\
  &= u
\end{align*}
\end{proof}

Next, we want to give general exact and approximate evaluations of the random coding error probability. To that end, let:
\begin{equation*}
  f(x,M) = 1-\BRA{1-x}^{M-1}
\end{equation*}
$f(x,M)$ represent the probability of decoding error given that $x$ is the pairwise error between the true codeword and another random metric of a randomly chosen codeword, and $M$ is the number of words in the codebook.
We have the following lemma that tightly bounds $f(x,M)$:
\begin{lemma} \label{lemma:upper_lower_bound}
  For any $x \in [0,1]$ and $M$: $\frac{1}{2} \cdot \min(1,M \cdot x) \leq f(x,M) \leq \min(1,M \cdot x)$
\end{lemma}
\begin{proof}
The upper bound follows from the union bound. The lower bound follows from a lemma by Shulman \cite[Lemma A.2]{ShulmanPHD} about the tightness of the clipped union bound for events that are pairwise independent.
\end{proof}

With this we have the following:
\begin{theorem} \label{theorem:achivble_bounds} Set $M=e^R$. \mynewline
\begin{equation}\label{PE:Exact}
  \bar{P}_{e,m,W}\BRA{R} = \EQWU{f(p_{e,m,X,Y}, M)}
\end{equation}
\begin{align} \label{PE:Upper}
            \bar{P}_{e,m,W}\BRA{R} &\leq \EQWU{\min\BRA{1,(M-1)\cdot p_{e,m,X,Y}} } \notag \\
                                   &\leq \EQW{ \min\BRA{1,(M-1)\cdot \ExU{ p_{e,m,X,Y} | \bx,\by} }}
\end{align}
\begin{equation}\label{PE:Lower}
  \bar{P}_{e,m,W}\BRA{R} \geq \frac{1}{2}\cdot \EQWU{\min\BRA{1,(M-1)\cdot p_{e,m,X,Y}} }
\end{equation}

  where $\EQWU{\cdot}$ is with respect to $Q(\bx)W(\by|\bx)$ and the uniform random variables ($\mathcal{U}_k$) in the random metric.
\end{theorem}
\begin{proof}
  \eqref{PE:Exact} follows because $p_{e,m,x,y}$ is the exact pairwise error probability given that $\bx$ was transmitted and $\by$ is the received word. \eqref{PE:Upper} and \eqref{PE:Lower} follow from lemma \ref{lemma:upper_lower_bound} and the concavity of $\min\BRA{1,(M-1)\cdot x}$ with respect to $x$.
\end{proof}
\begin{remark}
  The first bound in \eqref{PE:Upper}  is the \textbf{Clipped Union Bound}. It resembles the RCU bound \cite[Theorem 16]{polyanskiy2010channel}, with the twist about the tie breaking. The second bound is given in \cite{haim2013importance}.
\end{remark}

\section{General formula for the Clipped union bound} \label{section:general_formula}

Let $U(\bx,\by)$ be a positive random variable and define:
\begin{equation}\label{def:err_pu}
P_{U}(R) \triangleq \E{\min\BRA{1,e^{R}\cdot U(\bx,\by)}}.
\end{equation}
The random coding clipped union bound is just $P_{U}(R)$ with $U(\bx,\by)=p_{e,m,\bx,\by}$ and $M-1=e^R$.
To avoid cumbersome notation we will abbreviate and denote $P(R)$ instead of $P_{U}(R)$ and the R.V. $U$ will be held fixed.
The following lemma provides a formula for $P(R)$ in terms of the CDF of $-\log(U)$ and is useful for our purposes.

\begin{lemma}\label{lem:main_res}
Let $F(z)$ be the CDF of $-\log\BRA{U(\bx,\by)}$, \ie, $F(z) = \PR{-\log\BRA{U(\bx,\by)} \leq z}$. The probability is taken with respect to the expectation measure, \ie,\ $\PR{\mathcal{A}} = \E{\Ind{\mathcal{A}}}$. The following formula holds:
\begin{equation}\label{formulas:probability_error}
  P(R) = e^{R} \cdot \int_{R}^{\infty} F(z)\cdot e^{-z}dz
\end{equation}
If $F(z)$ is continuous at $z=R$, then the following also holds:
\begin{equation}\label{formulas:dev_pe}
\frac{\partial P(R)}{\partial R} = P(R)-F(R)
\end{equation}
\begin{equation}\label{formulas:dev_ee}
  \frac{\partial}{\partial R}\BRA{-\log\BRA{P(R)}} = \frac{F(R)}{P(R)}-1.
\end{equation}
\end{lemma}
\begin{proof}
The proof follows by using the expectation integral formula and changing the variables, and is given in the appendix.
\end{proof}
Define $E_r(R) = -\log\BRA{P(R)}$ as the random coding error exponent (non-asymptotic definition!). Combining theorem \ref{theorem:achivble_bounds} and lemma \ref{lem:main_res} we have a tight bound (up to factor 2) on the achievable rates by random coding at any rate. We summarize the performance of the random coding in the following:
\begin{theorem} \label{theorems:achive} \mynewline
\begin{equation}\label{RandomCoding:Bounds1}
  F(R) \leq P(R) = \frac{F(R)}{1+E_r'(R)}
\end{equation}
\begin{equation}\label{ErrorExponent:Bounds1}
    -1 \leq E_r(R) \leq 0.
\end{equation}
\end{theorem}
\begin{proof}
The proof also appears in the appendix.
\end{proof}

\section{General Converse} \label{section:general_converse}

From \eqref{RandomCoding:Bounds1}, $F(R)$ is a lower bound on the error probability of the clipped union bound for random codes. Taking the $\inf$ of $F(R)$ with respect to the prior $Q(\bx)$ gives a lower bound on the random coding error probability for all priors. In this section we prove that $F(R)$ is indeed a lower bound on the error probability of any fixed (non-random) code (up optimization on the prior). The result is similar to \cite[Lemma 1]{somekh2013general} but is given in terms of $F(R)$. Since we handle ties, the converse is on the performance of decoder where ties are broken arbitrary and not count automatically as error, as in \cite[Lemma 1]{somekh2013general}. For the matched case we prove that the bound matches the meta-converse of Polyanskiy \cite[Theorem 27]{polyanskiy2010channel}.

\subsection{Mismatched case}

\begin{theorem}[General mismatched converse] For any code with $M$ equiprobable codewords:
\begin{equation}\label{Converse:RD}
  \varepsilon  = F(R), 
\end{equation}
 where $Q(\bx)$ is distributed uniformly over the $M$ codewords and ties are broken arbitrarily. In particular
\begin{equation}\label{Converse:RD:Lower}
  \varepsilon  \geq \inf_{Q(\bx)} F(R),
\end{equation}
\end{theorem}
\begin{proof}

$Q(>) = Q(m(\bX,\by) > m(\bx,\by))$, $Q(=) = Q(m(\bX,\by) = m(\bx,\by))$

\begin{align*}
  \PR{E} &\overset{(a)}{=} \PR{E, Q(>) > 0} \\
         &+                \sum_{i=1}^{e^R} \PR{E, Q(>) = 0, Q(=) = i\cdot e^{-R} } \\
         &\overset{(b)}{=} \PR{Q(>) > 0} \cdot \PR{E| Q(>) > 0 } \\
         &+                \sum_{i=1}^{e^R} \PR{Q(>) = 0, Q(=) = i\cdot e^{-R} } \cdot \PR{E | Q(>) = 0, Q(=) = i\cdot e^{-R} }\\
         &\overset{(c)}{=} \PR{Q(>) \geq e^{-R}} \\
         &+                \sum_{i=2}^{e^R} \PR{Q(>) = 0, Q(=) = i\cdot e^{-R} } \cdot \frac{1}{i}\\
         &\overset{(c)}{=} \PR{Q(>) \geq e^{-R}} \cdot \PR{p_{e,m,\bx,\by} \geq e^{-R} | Q(>) \geq e^{-R}} \\
         &+                \sum_{i=1}^{e^R} \PR{Q(>) = 0, Q(=) = i\cdot e^{-R} } \cdot \PR{p_{e,m,\bx,\by} \geq e^{-R} | Q(>)=0, Q(=)=i\cdot e^{-R}} \\ 
         &= \PR{p_{e,m,\bx,\by} \geq e^{-R} } = F(R)
\end{align*}

\end{proof}

\subsection{Matched case}
In \cite{ElkayamEE2014} we evaluate the exponent of $F(R)$ for DMCs and got as a result the sphere packing bound. Polyanskiy \cite{polyanskiy2013saddle} evaluated the meta-converse and showed that asymptotically it also equals the sphere packing bound. A straight forward calculation of $F(R)$ for BSC and BEC (with uniform prior) gives the minimax converse lower bounds calculated in \cite{polyanskiy2013saddle}. In this section we will show that these bounds are indeed equivalent. Moreover, $F(R)$ will be equal to the power of the binary hypothesis test after optimization on $Q(\by)$ for any $Q(\bx)$.

The main result of this section is:
\begin{theorem}
\begin{equation}\label{FR:Equivalence}
  F(R) = \sup_{Q(y)} \beta_{1-e^{-R}}\BRA{Q(\bx)Q(\by), Q(\bx)W(\by|\bx)}.
\end{equation}
\end{theorem}

The proof outline is as follows:

\begin{enumerate}
  \item Let $A$ be the event $\BRAs{-\log(p_{e,m,x,y}) \leq R}$. Consider this event as a suboptimal test between $Q(\bx)Q(\by)$ and $Q(\bx)W(\by|\bx)$, \ie,\ $\delta(\bx,\by) = \Ind{-\log(p_{e,m,x,y}) \leq R} = \Ind{A}$
  \item When evaluated under $Q(\bx)Q(\by)$ for any $Q(\by)$ we will have $\PRs{Q(\bx)Q(\by),\mathcal{U}}{A} = 1-e^{-R}$ because of \eqref{RandomMetric:Uniform}. Thus $F(R) = \PRs{Q(\bx)W(\by|\bx),\mathcal{U}}{A} \geq \beta_{1-e^{-R}}\BRA{Q(\bx)Q(\by), Q(\bx)W(\by|\bx)}$ for all $Q(\by)$. This inequality also holds when we take the supremum over $Q(\by)$.
  \item For the reverse inequality, we show that there exists $Q(\by)$ that achieves $F(R)$. The reason is that the order between the $\bx's$ given $\by$ for the optimal test is exactly given by the order of the optimal metric $\log\BRA{W(\by|\bx)}$, which matches the order induced by $-\log(p_{e,m,x,y})$, and $Q(\by)$ is used to tune the threshold between all the different $\by's$.
\end{enumerate}

\ifFullProofs

\begin{proof}
Fix $Q(\by)$ and let $\delta(\bx,\by) = \Ind{-\log(p_{e,m,x,y}) \leq R}$ define a test between $Q(\bx)Q(\by)$ and $Q(\bx)W(\by|\bx)$. Compute the significance level associated with this test:
\begin{align*}
  \alpha    &= \PRs{Q(\bx)Q(\by),\mathcal{U}}{\delta(\bx,\by) = 1} \\
            &= \PRs{Q(\bx)Q(\by),\mathcal{U}}{-\log(p_{e,m,x,y}) \leq R)} \\
            &= \sum_y Q(\by) \cdot \PRs{Q(\bx),\mathcal{U}}{-\log(p_{e,m,x,y}) \leq R)} \\
            &= \sum_y Q(\by) \cdot \PRs{Q(\bx),\mathcal{U}}{p_{e,m,x,y} \geq 2^{-R})}   \\
            &\overset{(a)}{=} \sum_y Q(\by) \BRA{  1-2^{-R} } =  1-2^{-R},
\end{align*}
where (a) follows from \eqref{RandomMetric:Uniform}. This proves that the significance level of the test is $e^{-R}$ and:
\begin{align*}
   &\beta_{1-e^{-R}}\BRA{Q(\bx)Q(\by), Q(\bx)W(\by|\bx)} \\
   &\leq  \PRs{Q(\bx)W(\by|\bx),\mathcal{U}}{\delta(\bx,\by) = 1} \\
                                                        &= F(R).
\end{align*}
Taking the $\sup$ over $Q(\by)$ we have:
\begin{equation}\label{meta_eq_leq}
  \sup_{Q(\by)} \beta_{1-e^{-R}}\BRA{Q(\bx)Q(\by), Q(\bx)W(\by|\bx)} \leq F(R).
\end{equation}

To show the reverse inequality, it is sufficient to find $Q(\by)$ such that $\beta_{1-e^{-R}}\BRA{Q(\bx)Q(\by), Q(\bx)W(\by|\bx)} = F(R)$. Recall that we assumed the matched case, \ie\, $m(\bx,\by) = \log\BRA{W(\by|\bx)}$. For any $\by$ there exists $\bx_y$ such that: $Q(m(\bX,\by) > m(\bx_y,\by)) \leq 2^{-R} \leq Q(m(\bX,\by) \geq m(\bx_y,\by))$. Define $\tau_y \in [0,1]$ such that:
\begin{equation}\label{def:Threshold}
  2^{-R} = Q(m(\bX,\by) > m(\bx_y,\by)) + \tau_y \cdot Q(m(\bX,\by) = m(\bx_y,\by)).
\end{equation}

\begin{align*}
  F(R) &= \PRs{Q(\bx)W(\by|\bx),\mathcal{U}}{-\log(p_{e,m,x,y}) \leq R)} \\
       &= \PRs{Q(\bx)W(\by|\bx),\mathcal{U}}{p_{e,m,x,y} \geq 2^{-R})}   \\
       &= \sum_{\bx,\by} Q(\bx)W(\by|\bx)\PR{ p_{e,m,x,y} \geq 2^{-R} | \bx,\by}.
\end{align*}

We have 3 cases:
\begin{enumerate}
  \item When $m(\bx,\by) > m(\bx_y,\by)$ we have $p_{e,m,x,y} \leq Q\BRA{m(\bX,\by) \geq m(\bx,\by) } \leq Q\BRA{m(\bX,\by) > m(\bx_y,\by) } < 2^{-R}$. It follows that $\PR{ p_{e,m,x,y} \geq 2^{-R} | \bx,\by} = 0$ in this case. 
  \item When $m(\bx,\by) < m(\bx_y,\by)$ we have $p_{e,m,x,y} \geq Q\BRA{m(\bX,\by) > m(\bx,\by) } \geq Q\BRA{m(\bX,\by) \geq m(\bx_y,\by) } \geq 2^{-R}$. It follows that $\PR{ p_{e,m,x,y} \geq 2^{-R} | \bx,\by} = 1$ in this case.
  \item When $m(\bx,\by) = m(\bx_y,\by)$ we have $p_{e,m,x,y} = Q(m(\bX,\by) > m(\bx,\by)) + \mathcal{U} \cdot Q(m(\bX,\by) = m(\bx,\by))$. Since $Q(m(\bX,\by) > m(\bx,\by)) = Q(m(\bX,\by) > m(\bx_y,\by))$ and $Q(m(\bX,\by) = m(\bx,\by)) = Q(m(\bX,\by) = m(\bx_y,\by))$ we have:
      \begin{align*}
        &\PR{p_{e,m,x,y} \geq 2^{-R} |\bx,\by} \\
        &= \PR{p_{e,m,x_y,y} \geq 2^{-R} |\bx_y,\by} \\
                                 &= \PR{ Q(m(\bX,\by) > m(\bx_y,\by)) + \mathcal{U} \cdot Q(m(\bX,\by) = m(\bx_y,\by)) \geq 2^{-R} | \bx_y, \by}
      \end{align*}
      Plugging in \eqref{def:Threshold} and canceling out we have:
      \begin{align*}
        \PR{p_{e,m,x,y}|\bx,\by} &= \PR{\mathcal{U} \geq \tau_y } \\
                                 &= \BRA{1-\tau_y}.
      \end{align*}
\end{enumerate}
Combining all we get:

\begin{align*}
  F(R) &= \sum_{\bx,\by} Q(\bx)W(\by|\bx)\PR{ p_{e,m,x,y} \geq 2^{-R} | \bx,\by} \\
       &= \sum_{\bx,\by: m(\bx,\by) < m(\bx_y,\by) } Q(\bx)W(\by|\bx) + \sum_{\bx,\by: m(\bx,\by) = m(\bx_y,\by) } Q(\bx)W(\by|\bx)\BRA{1-\tau_y} \\
       &= \PRs{Q(\bx)W(\by|\bx)}{m(\bx,\by) < m(\bx_y,\by) } +\PRs{Q(\bx)W(\by|\bx)}{m(\bx,\by) = m(\bx_y,\by), \mathcal{U} \geq \tau_y }.
\end{align*}

In order to write $F(R)$ as the power of the binary hypothesis test, we want to have: $m(\bx_y,\by) = \log(W(\by|\bx_y)) = \log(Q(\by))+\eta$, which gives: $\eta = \sum_{y'} W(\by'|\bx_{y'})$ and $Q(\by) = \frac{W(\by|\bx_y)}{\sum_{y'} W(\by'|\bx_{y'})} = \eta^{-1}\cdot W(\by|\bx_y)$, and

\begin{align*}
  F(R) &= \PRs{Q(\bx)W(\by|\bx)}{m(\bx,\by) < m(\bx_y,\by) } +\PRs{Q(\bx)W(\by|\bx)}{m(\bx,\by) = m(\bx_y,\by), \mathcal{U} \geq \tau_y } \\
       &= \PRs{Q(\bx)W(\by|\bx)}{\log \BRA{W(\bx,\by)} < \log(Q(\by))+\eta } +\PRs{Q(\bx)W(\by|\bx)}{\log \BRA{W(\bx,\by)} = \log(Q(\by))+\eta, \mathcal{U} \geq \tau_y } \\
       &= \PRs{Q(\bx)W(\by|\bx)}{\log \BRA{\frac{W(\bx,\by)}{Q(\bx)}} < \eta } +\PRs{Q(\bx)W(\by|\bx)}{\log \BRA{\frac{W(\bx,\by)}{Q(\bx)}} = \eta, \mathcal{U} \geq \tau_y }.
\end{align*}
So indeed we see that $F(R)$ is the power of the Neyman-Pearson lemma for some $Q(\by)$ at significance level $e^{-R}$.
\end{proof}

\else
  \begin{proof}
  The proof appears in \cite{ElkayamFiniteSite} and is omitted here due to space limitation.
  \end{proof}

\fi

\section{Discussion and Further research} \label{section:discussion}

We conclude with several observations regarding the results we have obtained. First, we have shown the power of the random coding argument to achieve performance of the best code at a given rate up to a multiplicative factor that is given. The fundamental limit of the random coding argument is that the derivative of the random coding error exponent is bounded by -1, which prevents using this argument for rates where the ideal error exponent derivative is smaller than -1. Luckily, in general, the error exponent derivative for large enough rates is larger than -1.

In the general case, \ie,\ the mismatched decoding setting, the term $F(R)$ may be regarded as an ``operational'' quantity. However, the equivalence to the ``informational'' quantity \eqref{FR:Equivalence} for the matched case raises hope for an equivalent ``informational'' quantity in the mismatched case along the lines of \cite[Eq. 97-99]{somekh2013general}. We have several conjectures and this is still under investigation. Nevertheless, in the matched case, convexity and continuity properties are more easily derived using the binary hypothesis testing perspective; see \cite{polyanskiy2013saddle}. The validity of these properties to the mismatched metric are still under investigation.

Another cardinal difference between the matched and mismatched cases is the prior optimization. In the matched case, we have explicitly found the optimal output distribution $Q(\by)$ that achieves the $\sup$ in \eqref{FR:Equivalence} and the convexity of $F(R)$ as a function of the prior $Q(\bx)$, implies that the Karush $-$ Kuhn $-$ Tucker (KKT) conditions for the optimal prior should be computed. Moreover, channel symmetries (group $G$ act on the channel input and output, see \cite[Definition 3]{polyanskiy2013saddle}) can be used to assume uniform prior over the $G$-orbits of $Q(\bx)$, \ie, the prior is uniform on types. In the mismatched case, this should not be true anymore, as it is known that for the mismatch case the i.i.d. prior does not achieve the mismatch capacity.

Another interesting point is that the converse bound is given in terms of spheres in the $\bx$ domain. The intuitive sphere packing argument is used when packing the spheres in the $\by$ domain. The interpretation of $F(R)$ as a list decoding performance should also be further investigated.

In this paper we concentrated on the general bounds. Carrying out the calculation of $F(R)$ as a function of the block length, \ie\, $F_n(R)$, in the exponent (Large deviation), moderate deviation, and second order domains should provide the error probability performance of the coding problem in these regimes. Specifically, we should evaluate the convergence and convergence rates of $F_n(R)$. Even when $F_n(R)$ is known to converge (in some sense), the convergence of the derivative of the error exponent is not immediate.

\appendices

\section{Proofs} \label{App:AppendixA}

\subsection{Proof of lemma \ref{lem:main_res}}

We begin by evaluating the expectation using the integral formula \cite[Lemma 5.5.1]{koga2002information}:
\begin{align*}
  P(R)  &= \E{\min(1,e^{R} \cdot U(\bx,\by))} \\
        &= \int_0^{\infty} \PR{ \min(1,e^{R} \cdot U(\bx,\by)) \geq t}dt \\
        &\overset{(a)}{=} \int_0^{1} \PR{ e^{R} \cdot U(\bx,\by) \geq t}dt \\
        &= \int_0^{1} \PR{ -\log(U(\bx,\by)) \leq R-\log(t)}dt \\
        &= \int_0^{1} F(R-\log(t))dt
\end{align*}
where (a) follow because $\PR{ \min(1,e^{R} \cdot U(\bx,\by)) \geq t} = 0$ for $t > 1$ and $\PR{ \min(1,e^{R} \cdot U(\bx,\by)) \geq t} = \PR{ e^{R} \cdot U(\bx,\by) \geq t}$ for $t \leq 1$. Recall that we defined $F(z)=\PR{-\log(U(\bx,\by)) < z}$. Substituting $z=R-\log(t)$.
\begin{itemize}
  \item $e^z=\frac{e^{R}}{t}$, $t=e^{R-z}$
  \item $dz=-\frac{dt}{t}$, $dt=-t \cdot dz =-e^{R-z}\cdot dz$
  \item $t=0 \rightarrow z=\infty, t=1 \rightarrow z=R$
\end{itemize}
\begin{align*}
  P(R)  &= \E{\min(1,e^{R} \cdot U(\bx,\by))} \\
        &= \int_0^{1} F(R-\log(t))dt \\
        &= -\int_{\infty}^{R} F(z)\cdot e^{R-z}dz \\
        &= e^{R} \cdot \int_{R}^{\infty} F(z)\cdot e^{-z}dz \\
\end{align*}

which is \eqref{formulas:probability_error}.


\eqref{formulas:dev_pe} follow by differentiation under the integral sign \cite{flanders1973differentiation}
\footnote{ $\frac{d}{dx}\BRA{\int_{a(x)}^{b(x)} f(x,t)dt} = f(x,b(x))b'(x) - f(x,a(x))a'(x) + \int_{a(x)}^{b(x)} f_x(x,t)dt $}
of \eqref{formulas:probability_error}
\begin{align*}
  \frac{\partial P(R)}{\partial R} &= e^{R} \cdot \int_{R}^{\infty} F(z)\cdot e^{-z}dz \\
  &+ e^{R} \cdot \frac{\partial}{\partial R} \BRA{\int_{R}^{\infty} F(z)\cdot e^{-z}dz} \\
  &= e^{R} \cdot \int_{R}^{\infty} F(z)\cdot e^{-z}dz - e^{R} \cdot F(R) e^{-R} \\
  &= P(R) - F(R).
\end{align*}

And finally \eqref{formulas:dev_ee}
\begin{align*}
  &\frac{\partial}{\partial R}\BRA{- \log(P(R))} \\
  &= -\frac{1}{P(R)}\cdot\frac{\partial P(R)}{\partial R} \\
  &= -\frac{1}{P(R)}\cdot \BRA{P(R)-F(R)} \\
  &= \frac{F(R)}{P(R)}-1.
\end{align*}

\subsection{Proof of theorem \ref{theorems:achive}}

It follows from \eqref{formulas:probability_error} by using the monotony of the CDF function that:
\begin{align*}
  P(R)  &=    e^{R} \cdot \int_{R}^{\infty} F(z)\cdot e^{-z}dz \\
        &\geq e^{R} \cdot \int_{R}^{\infty} F(R)\cdot e^{-z}dz \\
        &=    e^{R} \cdot F(R)\cdot \int_{R}^{\infty} e^{-z}dz \\
        &= F(R).
\end{align*}
This proves that $F(R) \leq P(R)$. Since both $P(R), F(R)$ are positive we have $0 \leq \frac{F(R)}{P(R)} \leq 1$, which gives $-1 \leq E_r(R) \leq 0$.
Notice that $E_r(R) = -\log\BRA{P(R)}$ is the random coding error exponent and we can rewrite \eqref{formulas:dev_ee} as:
\begin{equation}\label{formulas:CDF_factor}
  P(R) = \frac{F(R)}{1+E_r'(R)}.
\end{equation}

\section{A ``sphere packing'' proof of the meta-converse} \label{App:AppendixB}

In this appendix we sketch a proof of the meta-converse based on a sphere packing argument. Let $Q(\by)$ a given distribution on $\cY$. For simplicity we assume that the decoder is deterministic. Later on we explain how this assumption can be removed. For each codeword $\bx_m$, $m=1..M$ let $P(\bx|m)$ denote the encoder channel input distribution given the $m-th$ message is sent. let $D_m$ denote the decoding region associated with the $m-th$ codeword and $\varepsilon_m = W(D_m^c|x_m)$ is the error probability associated with the $m-th$ codeword. Denote $1-\lambda_m = Q(D_m^c)$. Notice that $\sum_{m=1}^M \lambda_m = 1$ as the decoding regions cover the whole output space. Now:
\begin{align*}
  \varepsilon_m &= W(D_m^c|m) \\
                &= \sum_{\bx} W(D_m^c|\bx)P(\bx|m) \\
                &\overset{(a)}{\geq} \beta_{1-\lambda_m}\BRA{Q(\by)P(\bx|m); W(\by|\bx)P(\bx|m)}
\end{align*}
and:
\begin{align*}
  \varepsilon &= \frac{1}{M} \sum_{m=1}^M \varepsilon_m \\
              &\geq \frac{1}{M} \sum_{m=1}^M \beta_{1-\lambda_m}\BRA{Q(\by)P(\bx|m); W(\by|\bx)P(\bx|m)} \\
              &\overset{(b)}{\geq} \sum_{m=1}^M \beta_{1-\frac{1}{M} \sum_{m=1}^M \lambda_m}\BRA{\frac{1}{M} \sum_{m=1}^M Q(\by)P(\bx|m); \frac{1}{M} \sum_{m=1}^M W(\by|\bx)P(\bx|m)} \\
              &\geq \sum_{m=1}^M \beta_{1-\frac{1}{M}}\BRA{Q(\by)P(\bx); W(\by|\bx)P(\bx)}
\end{align*}
where $P(\bx) = \frac{1}{M} \sum_{m=1}^M P(\bx|m)$. (a) by definition of $\beta_{\alpha}$ and the non-optimality of the test $D_m^c$, (b) is the convexity of $(\alpha, P(\bx)) \rightarrow \beta_{\alpha}(P(\by|\bx)P(\bx), Q(\by|\bx)P(\bx))$ \cite[Theorem 6]{polyanskiy2013saddle}.
A randomized decoder is given by the output distribution $g(\bar{m}|\by)$. We can simulate this choice deterministically given an independent uniform random variable. Hence the randomized decoder can be modeled as a deterministic decoder on the output of the channel and an auxiliary independent random variable. It is easy to see that this auxiliary variable does not change the proof above. Like always, take the $\sup$ on $Q(\by)$ to get the largest lower bound and $\inf$ on $P(\bx)$ to get a code independent bound.

\bibliographystyle{IEEEtran}
\bibliography{bib}

\end{document}

%% file: common_defs.tex
\usepackage{times}
\usepackage{hyperref} 
\usepackage{amsmath}
\usepackage{amssymb}
\usepackage{amsthm}
\usepackage{bbm}
\usepackage{dsfont}

\usepackage{epstopdf} 
\usepackage{graphicx}

\newtheorem{theorem}{Theorem}
\newtheorem{lemma}{Lemma}
\newtheorem{proposition}[theorem]{Proposition}

\theoremstyle{definition}
\newtheorem{definition}{Definition}
\theoremstyle{remark}
\newtheorem{remark}{Remark}

\newcommand{\mynewline}{\mbox{}\\}
\newcommand{\E}[1]{\mathbb{E} \left( #1 \right)}

\newcommand{\EQW}[1]{\mathbb{E}_{ \bX, \bY} \left( #1 \right)}
\newcommand{\EQWU}[1]{\mathbb{E}_{\bX, \bY, \mathcal{U}} \left( #1 \right)}
\newcommand{\ExU}[1]{\mathbb{E}_{\mathcal{U}} \left( #1 \right)}

\newcommand{\BRA}[1]{\left( #1 \right)}

\newcommand{\BRAs}[1]{\left\{ #1 \right \}}
\newcommand{\PR}[1]{Pr\left\{ #1 \right\}}
\newcommand{\PRs}[2]{Pr_{#1}\left\{ #2 \right\}}

\newcommand{\bX}{{\textbf X}}
\newcommand{\bY}{{\textbf Y}}

\newcommand{\cX}{{\mathcal X}}
\newcommand{\cY}{{\mathcal Y}}

\newcommand{\cC}{{\mathcal C}}
\newcommand{\cD}{{\mathcal D}}

\newcommand{\bx}{\textbf{x}}
\newcommand{\by}{\textbf{y}}

\newcommand{\ie}{{\emph{i.e.}}}
\newcommand{\eg}{{\emph{e.g.}}}

\newcommand{\etal}{{\emph{et al.}}}

\newcommand{\Ind}[1]{ \mathds{1}_{\BRAs{#1}} }

\newcommand{\MIN}[1]{ \smash{\displaystyle\min_{#1}} }



%% file: FiniteBlockLen.bbl
\begin{thebibliography}{1}
\providecommand{\url}[1]{#1}
\csname url@samestyle\endcsname
\providecommand{\newblock}{\relax}
\providecommand{\bibinfo}[2]{#2}
\providecommand{\BIBentrySTDinterwordspacing}{\spaceskip=0pt\relax}
\providecommand{\BIBentryALTinterwordstretchfactor}{4}
\providecommand{\BIBentryALTinterwordspacing}{\spaceskip=\fontdimen2\font plus
\BIBentryALTinterwordstretchfactor\fontdimen3\font minus
  \fontdimen4\font\relax}
\providecommand{\BIBforeignlanguage}[2]{{%
\expandafter\ifx\csname l@#1\endcsname\relax
\typeout{** WARNING: IEEEtran.bst: No hyphenation pattern has been}%
\typeout{** loaded for the language `#1'. Using the pattern for}%
\typeout{** the default language instead.}%
\else
\language=\csname l@#1\endcsname
\fi
#2}}
\providecommand{\BIBdecl}{\relax}
\BIBdecl

\bibitem{polyanskiy2010channel}
Y.~Polyanskiy, H.~V. Poor, and S.~Verd{\'u}, ``Channel coding rate in the
  finite blocklength regime,'' \emph{Information Theory, IEEE Transactions on},
  vol.~56, no.~5, pp. 2307--2359, 2010.

\bibitem{DBLP:journals/tit/VerduH94}
S.~Verd{\'u} and T.~S. Han, ``A general formula for channel capacity,''
  \emph{IEEE Transactions on Information Theory}, vol.~40, no.~4, pp.
  1147--1157, 1994.

\bibitem{somekh2013general}
A.~Somekh-Baruch, ``A general formula for the mismatch capacity,'' \emph{arXiv
  preprint arXiv:1309.7964}, 2013.

\bibitem{ElkayamEE2014}
\BIBentryALTinterwordspacing
N.~Elkayam and M.~Feder. (2014) Exact evaluation of the random coding error
  probability and error exponent. [Online]. Available:
  \url{http://www.eng.tau.ac.il/~elkayam/ErrorExponent.pdf}
\BIBentrySTDinterwordspacing

\bibitem{haim2013importance}
E.~Haim, Y.~Kochman, and U.~Erez, ``The importance of tie-breaking in
  finite-blocklength bounds,'' in \emph{Information Theory Proceedings (ISIT),
  2013 IEEE International Symposium on}.\hskip 1em plus 0.5em minus 0.4em\relax
  IEEE, 2013, pp. 1725--1729.

\bibitem{ShulmanPHD}
N.~Shulman, ``Communication over an unknown channel via common broadcasting,''
  Ph.D. dissertation, Tel Aviv University, 2003.

\bibitem{polyanskiy2013saddle}
Y.~Polyanskiy, ``Saddle point in the minimax converse for channel coding,''
  \emph{Information Theory, IEEE Transactions on}, vol.~59, no.~5, pp.
  2576--2595, 2013.

\bibitem{koga2002information}
H.~Koga \emph{et~al.}, \emph{Information-spectrum methods in information
  theory}.\hskip 1em plus 0.5em minus 0.4em\relax Springer, 2002, vol.~50.

\bibitem{flanders1973differentiation}
H.~Flanders, ``Differentiation under the integral sign,'' \emph{American
  Mathematical Monthly}, pp. 615--627, 1973.

\end{thebibliography}
